%% file: arxiv-v1.tex
\documentclass[11pt,a4paper]{article}

\usepackage{fullpage}
\usepackage{amsmath,amsfonts,amssymb}
\usepackage{hyperref}

\usepackage[noend]{algpseudocode}
\usepackage{algorithm}
\usepackage{color}
\usepackage[noadjust]{cite}
\usepackage{tikz}
\usepackage{url}
\usepackage{pgfplots}
\usepackage{blindtext}
\pgfplotsset{compat=1.5}
\newcommand{\tikzscale}{1}

\newtheorem{theorem}{Theorem}
\newtheorem{definition}{Definition}
\newtheorem{lemma}{Lemma}
\newtheorem{fact}{Fact}
\newenvironment{proof}{\medskip\textit{Proof.}}{\hfill$\square$}

\newcommand{\bc}{\mathsf{bc}}
\newcommand{\lc}{\mathsf{lc}}
\newcommand{\out}{\mathsf{out}}
\DeclareMathOperator{\dist}{\mathsf{dist}}
\DeclareMathOperator{\diam}{\mathsf{diam}}
\DeclareMathOperator{\Diam}{\mathsf{Diam}}
\DeclareMathOperator{\minhop}{\mathsf{minhop}}
\DeclareMathOperator{\maxhop}{\mathsf{maxhop}}
\renewcommand{\deg}{\mathsf{deg}}
\newcommand{\nh}{\mathsf{NH}}
\newcommand{\varnh}{\mathrm{NH}}
\newcommand{\ph}{\mathsf{PH}}
\newcommand{\varph}{\mathrm{PH}}
\newcommand{\congest}{\textsc{congest}}

\begin{document}

\title{Simple and Fast Distributed Computation of \\ Betweenness Centrality}

\author{
	Pierluigi Crescenzi\thanks{On leave from DiMaI, Universit\`a degli Studi di Firenze, I-50134 Firenze, Italy.} \\
	{\small Universit\'e de Paris}\\ {\small IRIF, CNRS}
	\and
	Pierre Fraigniaud\thanks{Supported by ANR projects DESCARTES and FREDA and INRIA project GANG.} \\ 
	{\small Universit\'e de Paris}\\ {\small IRIF, CNRS}
	\and 
	Ami Paz\thanks{Supported by the Fondation Sciences Math\'ematiques de Paris (FSMP).} \\ 
	{\small Faculty of Computer Science}\\ {\small Univarsity of Vienna}
}

\date{}

\maketitle

\begin{abstract}
Betweenness centrality is a graph parameter that has been successfully applied to network analysis. In the context of computer networks, it was considered for various objectives, ranging from routing to service placement. However, as observed by Maccari et al. [INFOCOM 2018], research on betweenness centrality for improving protocols was hampered by the lack of a usable, fully distributed algorithm for computing this parameter. We resolve this issue by designing an efficient algorithm for computing betweenness centrality, which can be implemented by minimal modifications to any distance-vector routing protocol based on Bellman-Ford. The convergence time of our implementation is shown to be proportional to the diameter of the network. 
\end{abstract}

\section{Introduction}

Betweenness centrality~\cite{Freeman97} is a measure of ``importance'' attributed to every node of a graph. Roughly, the betweenness centrality of a node~$v$ is the sum, taken over all pairs $(s,t)$ of source-target nodes, of the ratio between the number of shortest paths from $s$ to $t$ passing through~$v$, and the total number of shortest paths from $s$ to $t$. Thus, a node with high betweenness centrality belongs to relatively many shortest paths, while a node with low betweenness centrality belongs to relatively few shortest paths. Betweenness centrality has been successfully applied to network analysis: In social networks, a node with high betweenness centrality is plausibly an influential node; in computer networks, a node with high betweenness centrality might be efficiently used  for storing relevant resources, but may also cause severe damage to the communications in case of failure or malfunction. 

Consequently, betweenness centrality and its variants have been used for optimizing the behavior of communication networks and computer networks~\cite{KDT2010}, whether it is for wireless mesh networks design~\cite{KasAW+12}, routing~\cite{DolevEP10}, link-sensing~\cite{MaccariC16}, resource placement~\cite{PantazopoulosKS14} and allocation~\cite{ZilbermanPE17}, topology control~\cite{RodasL15}, transmission rates optimization~\cite{BaldesiMC17}, or security~\cite{MaccariC14}. For instance, the optimal frequency at which the incident links must be sensed at each node is known to be inversely proportional to the square root of the betweenness centrality of the node~\cite{MaccariC16}.

Nevertheless, as observed by Maccari et al.~\cite{MaccariGGMC18}, network optimization techniques based on betweenness centrality suffer from two main issues. First, even if every node has access to information about the whole network, which is the case in link-state protocols, the computation of the betweenness centrality may require excessive computational resources; hence, various heuristics and random sampling techniques were proposed in order to reduce the computation time~\cite{BaderKMM07,BaglioniGPL12,BrandesP07,GeisbergerSS08,JacobKL+05,Lim2011,MaccariNC16,MaiyaB10,MOG19,PuzisEZDB15,RiondatoK16,RiondatoU18}. Second, there are no known efficient algorithms for computing betweenness centrality in the context of distance-vector protocols. Existing distributed algorithms for computing betweenness centrality or all-pairs shortest paths are either designed for models that are too weak compared to real-world networks supporting distance-vector protocols (e.g., \congest\/ model)~\cite{HoangPDGYPR19}, or dedicated to restricted classes of network topology (e.g., DAGs or trees)~\cite{YouTQ17,WT13,WT14}, or they exchange an amount of information between nodes that exceed the capacity of distance-vector protocols~\cite{ZK03}. Actually, even the elegant and practical algorithm by Maccari et al.~\cite{MaccariGGMC18} for computing a related measure called load centrality~\cite{Brandes08,GKK01} exchanges slightly more information between nodes than one would expect from a distance-vector protocol. 

\subsection{Our Results}

We describe a simple and fast distributed algorithm for computing betweenness centrality. Specifically, our algorithm enables every node~$v$ to compute its own centrality~$\bc_v$. 

Our algorithm is simple in the sense that it can be implemented by minimal modifications to distance-vector protocols based on Bellman-Ford. Concretely, Bellman-Ford asks every node~$v$ to send to each of its neighbors a pair of values $(t,d)$ for every target-node~$t$, where $d$ is the current distance from $v$ to~$t$, as perceived by~$v$. Our algorithm simply asks every node to send to each neighbor a quadruple of values $(t,d,s,b)$ for every target-node~$t$, where $s$ is the current estimation at $v$ of the number of shortest paths from $v$ to $t$, and $b$ is the current contribution of $t$ to the betweenness centrality of~$v$. 

Our algorithm is fast in the sense that it converges in a number of distance-vector phases proportional to the diameter of the network. Moreover, the amount of computations performed at each node~$v$ upon reception of a message from a neighbor~$u$ related to a target~$t$ is (amortized) constant, i.e., independent of the size of the network.  

We have performed an extensive set of simulations confirming both the correctness analysis of our algorithm, and its efficiency. We have considered different scenarios, including weighted and unweighted networks, and various topologies generated by synthetic models (grids, Erd\"os-R\'enyi, etc.) or extracted from real-world networks. 

The main outcome of this paper is that betweenness centrality can be efficiently computed at every node in a distributed manner, even in the context of distance-vector protocols. As a consequence, there is no obstacle for using betweenness centrality for optimizing the functionality of networks, as far as computing this parameter at run-time by the network itself is concerned. This resolves a question left open in the work of Maccari et al.~\cite{MaccariGGMC18}.

\subsection{Related Work}
\label{subsec:RW}

In addition to the aforementioned contributions to computing or approximating betweenness centrality, significant effort has been made by the distributed computing community for efficiently computing related graph measure, the main one being all-pairs shortest paths (APSP). The vast majority of the results in the distributed setting were however derived under the so-called \congest\/ model~\cite{Peleg2000}, in which failure-free processing nodes perform  in a sequence of synchronously rounds, and the message traversing any link at any round must carry a constant number of words only  (i.e., the messages are of $O(\log n)$ bits in $n$-node networks). In contrast, distance-vector protocols can be viewed as exchanging $\Theta(n)$ words along each edge at each phase of communication between neighbors, each word consisting of a pair $(t,d)$ of values. Nevertheless, it is worth mentioning previous contributions~\cite{HoangPDGYPR19,PR2018,HuaFAQLSJ2016}, computing betweenness centrality in unweighted graphs (directed or not) in $O(n)$ rounds under the \congest\/ model. Some of these results also apply to weighted digraphs, but to DAGs only~\cite{PR2018}.

Computing all-pairs shortest paths in the \congest\/ model was intensively studied. In unweighted graphs, several $O(n)$-round algorithms exist~\cite{HolzerW2012,PelegRT2012,LenzenP2013}, and these were improved by an $O(n/\log n)$-rounds algorithm~\cite{HuaFQALSJ2016},
which is tight~\cite{FrischknechtHW2012}.
In the weighted case, an $\tilde O(n)$-rounds algorithm was recently presented~\cite{BernsteinN2018},
almost matching the $\Omega(n)$ lower bound~\cite{CensorKP2017}.

Finally, the computation of betweenness centrality was studied in the context of dynamic graphs~\cite{BergaminiMS15,BergaminiM15,HayashiAY15,KourtellisMB15}. These works makes it possible to maintain the betweenness centralities of the nodes in a changing system, without having to recompute everything from scratch when a change occurs (e.g., adding or removing an edge). However, these algorithms are centralized and not distributed. 
Note that in the distance-vector model, maintaining even just the distances in a changing network (APSP) is far from being trivial (see, e.g.,~\cite[Section 4.2.2]{PetersonD2013}).

\section{Definitions}

Let $G=(V,E)$ be a connected undirected graph with positive edge-weights. The weight of an edge $e\in E$ is denoted by $w(e)>0$. A path between two distinct nodes $s,t\in V$ is a sequence $v_0,\dots,v_k$ with $k\geq 0$, $v_0=s$, $v_k=t$, and $\{v_{i-1},v_i\}\in E$ for every $i=1,\dots,k$. The length of such a path is 
$
\sum_{i=1}^k w(\{v_{i-1},v_i\}). 
$
A path with minimum length between $s$ and~$t$ is called a shortest path between $s$ and~$t$. The distance $\dist(s,t)$ between two nodes $s$ and $t$ is the length of a shortest path between $s$ and~$t$. The weighted diameter of $G=(V,E)$ is defined as 
$
\max_{s,t\in V}\dist(s,t).
$
The weighted diameter is however not reflecting the convergence time of routing protocols based on Bellman-Ford. The complexity of the latter is indeed related to the hop-diameter of~$G$ defined as follows. For any two nodes $s\neq t$, let $P_1,\dots,P_\ell$ be the $\ell\geq 1$ shortest paths between $s$ and $t$ in~$G$. Let $\minhop(s,t)$ be the minimum, taken over all $i=1,\dots,\ell$, of the number of edges of~$P_i$. The hop-diameter of $G$ is set as
\[
\diam(G)=\max_{s,t\in V}\minhop(s,t).
\]
The standard distributed version of Bellman-Ford enables every node $v$ to compute $\dist(v,t)$ for every $t\in V$ --- see  Algorithm~\ref{algoBF}. Here and later, we use different notations to distinguish the defined value (e.g. $\dist$), and the value computed by the algorithm (e.g. $D$). This algorithm converges in $\diam(G)$ phases, where a phase is defined as the time required by every node~$v$ to send all the messages $(t,D[t])$, $t\in V$, to all its neighbors, receive all the messages $(t,d)$, $t\in V$, sent by each of its neighbors, and process all these messages. Indeed, a simple induction on $h\geq 0$ enables to show that, for every $t\in V$, every node $v$ with $\minhop(v,t)\leq h$ has computed $\dist(v,t)$ correctly after $h$ rounds. 

\begin{algorithm}[tb]
\caption{Distributed Bellman-Ford at node $v$}
\label{algoBF}
\begin{algorithmic}[1]

\Function{init}{} 
\State \textbf{forall} $t\in V$  \textbf{do} $D[t] \leftarrow +\infty$ \Comment{\emph{$D$ is distance vector}}
\State $D[v] \leftarrow 0$ \Comment{\emph{$\dist(v,v)=0$}} 
\EndFunction

\medskip

\Function{send}{} 
\Loop \Comment{\emph{periodic updates are sent to neighbors}}
	\State \textbf{forall} $t\in V$ \textbf{do} send $(t,D[t])$ to every neighbor $u$
\EndLoop	 
\EndFunction 
	
\medskip

\Function{receive}{message $(t,d)$ from neighbor $u$} 
\If{$d+w(\{u,v\})<D[t]$} 
	\State $D[t]\leftarrow d+w(\{u,v\})$
\EndIf
\EndFunction 
\end{algorithmic}
\end{algorithm}

For any two vertices $s\in V$ and $t\in V$, let $\sigma_{s,t}$ denote the number of shortest paths from $s$ to $t$ in~$G$ (with $\sigma_{s,s}=1$), and let $\sigma_{s,t}(v)$ denote the number of shortest paths from $s$ to $t$ passing through node~$v$ (with $\sigma_{s,s}(s)=1$).

\begin{definition}\label{def:bc}
The \emph{betweenness centrality}~\cite{Freeman97} of node $v$ is 
\[
\bc_v=\frac{1}{(n-1)(n-2)}\sum_{s\neq v, t\neq v}\frac{\sigma_{s,t}(v)}{\sigma_{s,t}}. 
\]
\end{definition}

Note that, since $G$ is connected, $\sigma_{s,t}>0$ for every two vertices $s,t\in V$, and thus $\bc_v$ is well defined for every $v\in V$. 

As mentioned earlier in the text, betweenness centrality is not only a central notion in the context of network analysis, but can also be used for various optimization scenarios in the context of computer networks. In particular, Maccari and Cigno~\cite{MaccariC16} focused on optimizing the frequency of \textsc{hello} messages for link-sensing in wireless networks, and showed that the optimal frequency $f(v)$ at which every node~$v$ must sense its neighbors is 
\begin{equation}\label{eq:freq}
f(v) \approx \sqrt\frac{\deg_v}{\bc_v},
\end{equation}
where $\deg_v$ denotes the degree of~$v$, i.e., its number of neighbors. Unfortunately, the above formula can hardly be practically used in absence of an efficient way of computing $\bc_v$ at every node~$v$. To overcome this,  Maccari et al.~\cite{MaccariGGMC18} have proposed to replace the use of betweenness centrality by the use of load centrality~\cite{Brandes08}, defined as follow. 

Assume that one unit of flow is pushed from $s$ to $t$ in~$G$ according to the following rule. For any node~$v$, let $\out_{s,t}(v)$ denote the set of edges $e=\{v,v'\}$ incident to $v$ such that there exists a shortest path from $s$ to $t$ traversing the edge $e$ from $v$ to $v'$. If a fraction $r$ of the flow from $s$ to $t$ is received by a node $v\neq t$, and if $|\out_{s,t}(v)|=k>0$, then $v$ forwards a fraction $r/k$  of flow along each of the edges in $\out_{s,t}(v)$. Let $\theta_{s,t}(v)$ denote the amount of flow from $s$ to $t$ traversing~$v$. 

\medskip

\begin{definition}
The \emph{load centrality}~\cite{Brandes08} of node~$v$ is  
\[
\lc_v=\sum_{s,t}\theta_{s,t}(v). 
\]
\end{definition}

A distributed algorithm for computing load centrality has been described and analyzed in~\cite{MaccariGGMC18}. It was shown to be implementable by minimal modifications to the Bellman-Ford algorithm. Roughly, instead of every node $v$ sending just a pair $(t,D[t])$ for every node~$t$, where $D[t]$ is the current estimation of the distance between $v$ and~$t$ as perceived by~$v$, every node~$v$ sends a tuple 
\[
(t,D[t], \nh[t], \ell, L[t])
\]
where $\nh[t]$ is the list of ``next hops'' to~$t$, i.e., the set of neighbors of~$v$ on a shortest path from $v$ to~$t$ (corresponding to $\out_{v,t}(v)$), $\ell$ is the overall flow passing through $v$ to reach $t$, and $L[t]$ is the load centrality of $t$, as far as $v$ knows. Note that this tuple can be significantly larger than a pair $(t,D[t])$ as $\nh[t]$ can potentially contain many entries. Instead, our algorithm for computing the betweenness centrality exchanges messages with a bounded number of entries. Note also that the algorithm in~\cite{MaccariGGMC18} converges in a number of phases related to
\[
\Diam(G)=\max_{s,t\in V}\maxhop(s,t), 
\]
where $\maxhop(s,t)$ is the maximum, taken over all shortest paths between $s$ and $t$, of the number of edges of each path. 
(Recall that the hop-diameter of $G$ is defined as $\diam(G)=\max_{s,t\in V}\minhop(s,t)$.)
Running for $\Diam(G)$ rounds seems unavoidable, by the definition of the load centrality, because a positive fraction of the flow from $s$ to $t$ is indeed shipped via a shortest path of length $\maxhop(s,t)$ between $s$ and~$t$. Nevertheless~\cite{MaccariGGMC18} the practical performance of the algorithm computing load centrality remains close to $\diam(G)$ --- this is mainly due to the fact that there are typically few shortest paths between any two nodes in real-world weighted networks, all with very similar numbers of edges.

In the next section, we show that betweenness centrality can also be computed distributively, by minimal modifications of Bellman-Ford.  In particular, our algorithm enables to compute exactly the optimal frequency $f(v)$ of each node~$v$, as described in Eq.~\eqref{eq:freq}. The performance of our algorithm is also shown to be close to the diameter of the network.

\section{Distributed Computation of BC}

This section describes our distributed algorithm for computing betweenness centrality. This algorithm is in essence a distributed implementation of the dynamic programming algorithm by Brandes~\cite{Brandes01}. Note that this latter algorithm is well suited for a distributed implementation because a node $v$ does not need to know the number of shortest paths from $s$ to $t$ in order to compute the contribution of the pair $s$ and $t$ to its betweenness centrality (see Lemma~\ref{lem:recursive-formula-bc}).

For the ease of presentation, we view the undirected graph $G=(V,E)$ as a directed graph where every edge $\{u,v\}$ is replaced by two symmetric arcs $(u,v)$ and $(v,u)$, both with the same weight as the edge  $\{u,v\}$. Also, we extend the definition of $\sigma_{s,t}(v)$ from nodes to arcs, by denoting, for every arc $(u,v)\in E$, $\sigma_{s,t}(u,v)$ as the number of shortest paths from $s$ to $t$  traversing the arc $(u,v)$, i.e., traversing the edge $\{u,v\}$ from $u$ to~$v$. The following two facts directly follow from the definitions. 

\medskip

\begin{fact}\label{fact:svvt}
If $\sigma_{s,t}(v)\neq 0$, i.e., if $v$ belongs to a shortest path from $s$ to $t$, then $\sigma_{s,t}(v)=\sigma_{s,v}\cdot\sigma_{v,t}$. Similarly, if the arc $(u,v)$ belongs to a shortest path from $s$ to $t$, then $\sigma_{s,t}(u,v)=\sigma_{s,u}\cdot\sigma_{v,t}$. 
\end{fact}

\medskip

Let $v\in V$. For every $t\in V$, let $\nh_v(t)$ be set of neighbors of $v$ that belong to some shortest paths from $v$ to $t$, and, for every $s\in V$, let $\ph_v(s)$ be the set of neighbors $u$ of $v$ such that $v$ belongs to some shortest paths from $s$ to $u$. ($\nh$ and $\ph$ holds for ``next hop'' and ``previous hop'', respectively). Note that, since $G$ is undirected, for every $u, v, w\in V$, $u\in \ph_v(w)$ if and only if $v\in \nh_u(w)$.

\medskip

\begin{fact} \label{fact:sigmaphs}
For every $t\neq v$, we have $\sigma_{v,t}=\sum_{u\in\nh_v(t)}\sigma_{u,t}$. If $v$ is on a shortest path from $s$ to $t$, then $\sigma_{v,t}=\sum_{u\in\ph_v(s)}\sigma_{v,t}(v,u)$. 
\end{fact}

\medskip

Finally, for every $s\in V$, let $\bc_v(s)$ be the contribution of the source $s$ to $\bc_v$, that is, 
\[
\bc_v(s)=\sum_{t\neq v} \frac{\sigma_{s,t}(v)}{\sigma_{s,t}}. 
\]
By definition, we have 
\[
\bc_v=\frac{1}{(n-1)(n-2)}\sum_{s\neq v}\bc_v(s)
\]
for every node~$v$. The following result is central in our distributed implementation of Brandes algorithm. 

\medskip
\begin{lemma}[{\cite[Theorem 6]{Brandes01}}]\label{lem:recursive-formula-bc}
For every $s\neq v$, we have 
\[
\bc_v(s)=\sigma_{s,v} \sum_{u\in\ph_v(s)}\frac{\bc_u(s)+1}{\sigma_{s,u}}.
\]
\end{lemma}

\medskip

Our algorithm requires that each node $v\in V$ computes and maintains the sets $\nh_v(t)$ and $\ph_v(t)$ for every node $t\in V$.  This is achieved by applying simple modifications to the function \textsc{receive} in Bellman-Ford algorithm, as described in Algorithm~\ref{algoBC}, Lines~\ref{lineNHbegin}-\ref{lineNHend}. In this latter algorithm, and in the remaining of the text, we denote by $N(v)$ the set of neighbors of $v\in V$ in $G=(V,E)$, that is, 
\[
N(v)=\{u\in V: \{u,v\}\in E\},  
\]
and we denote by $N[v]=N(v)\cup\{v\}$  the closed neighborhood of node~$v$. 

\medskip

\begin{algorithm}[tb]
\caption{Computing betweenness centrality at node $v$}
\label{algoBC}
\begin{algorithmic}[1]

\Function{init}{} 
\ForAll {$t\in V$} 
	\State $D[t] \leftarrow +\infty$ \Comment{\emph{$D$ is a distance vector}}
	\State $\varnh[t] \leftarrow \emptyset$ \Comment{\emph{next-hop vector}}
	\State $\varph[t] \leftarrow \emptyset$ \Comment{\emph{previous-hop vector}}
	\ForAll {$u\in N[v]$} 
		\State $B[u,t] \leftarrow 0$ \Comment{\emph{eventually $\bc_u(t)$}}
		\State $S[u,t] \leftarrow 0$ \Comment{\emph{eventually $\sigma_{u,t}$}}
	\EndFor
\EndFor
\State $S[v,v] \leftarrow 1$ \Comment{\emph{$\sigma_{v,v}=1$}}
\State $D[v] \leftarrow 0$ \Comment{\emph{$\dist(v,v)=0$}}

\EndFunction

\medskip

\Function{send}{} 
\Loop \Comment{\emph{periodic updates are sent to neighbors}}
	\ForAll {$t\in V$} 
		\State send $(t,D[t],S[v,t],B[v,t])$ to all $u\in N(v)$
	\EndFor
\EndLoop	 
\EndFunction 
\medskip

\Function{receive}{message $(t,d,s,b)$ from $u\in N(v)$} 
\State $\varnh[t]\leftarrow \varnh[t] \smallsetminus \{u\}$ \Comment{\emph{initialization: $u$ is removed}}\label{lineNHbegin}
\State $\varph[t]\leftarrow \varph[t] \smallsetminus \{u\}$ \Comment{\;\;\;\;\;\;\emph{from both $\varnh$ and $\varph$}}
\If{$d+w(\{u,v\})<D[t]$}
	\State $D[t]\leftarrow d+w(\{u,v\})$ \Comment{\emph{Bellman-Ford update}}
\ElsIf{$d+w(\{u,v\}) = D[t]$} 
	\State $\varnh[t]\leftarrow \varnh[t] \cup \{u\}$ \Comment{\emph{$u$ is placed (back) in $\varnh$}}
\ElsIf{$d - w(\{u,v\}) = D[t]$}
	\State $\varph[t]\leftarrow \varph[t] \cup \{u\}$ \Comment{\emph{$u$ is placed (back) in $\varph$}}
\EndIf \label{lineNHend}
\State $S[u,t] \leftarrow s$ \label{lineBCbegin}
\State $B[u,t] \leftarrow b$ \label{line:testvv}
\State \textbf{if} $t\neq v$ \textbf{then}  $S[v,t]\leftarrow \sum_{x\in \varnh[t]}S[x,t]$ \label{line:computeS}
\State $B[v,t]\leftarrow S[v,t]\cdot \sum_{x\in \varph[t]}\frac{B[x,t]+1}{S[x,t]}$ \label{lineBCend}
\State $C \leftarrow \sum_{x \neq v} B[v,x]$ \Comment{\emph{eventually $C=\bc_v$}} \label{lineBC}
\EndFunction 
\end{algorithmic}
\end{algorithm}

\begin{lemma} \label{lem:NHetPH}
Algorithm~\ref{algoBC} enables every node~$v\in V$ to compute  the sets $\nh_v(t)$ and $\ph_v(t)$ for every node $t\in V$.  More specifically, at node~$v$, for every node~$t\in V$, $\varnh[t]=\nh_v(t)$ after $\maxhop(v,t)+1$ phases, and $\varph[t]=\ph_v(t)$ after $\Diam(G)+2$ phases.
\end{lemma}

\begin{proof} 
Let $v\in V$. The proof is by induction on $h=\maxhop(v,t)$. For $h=0$, the fact that $D[v]$ remains zero throughout the execution of the algorithm guarantees that $\varnh[v]$ remains empty throughout the execution, as desired. Let $h>0$, and let us assume that, for every $t\in V$ with $\maxhop(v,t)<h$, $\varnh[t]=\nh_v(t)$ after $h$ phases. Let $t\in V$ with $\maxhop(v,t)=h$. Since $\minhop(x,y)\leq \maxhop(x,y)$ for every two nodes $x,y\in V$, we have $D[t]=\dist(v,t)$ after $h$ phases. Let $u\in \nh_v(t)$, and let us consider the updates occurring at phase $h+1$. Since $\maxhop(u,t)<\maxhop(v,t)$, the value $d$ in the tuple $(t,d,s,b)$ sent by $u$ to $v$ at phase $h+1$ satisfies $d=\dist(u,t)$. Therefore, the equality $d+w(\{u,v\}) = D[t]$ holds at~$v$, resulting to $u$ being added to $\varnh[t]$, as desired. That is, for every $u\in \nh_v(t)$, Algorithm~\ref{algoBC} guarantees that $u\in\varnh[t]$ at $v$, after $h+1$ phases. This completes the proof of the induction for the next hops vector. 

The proof is similar for the previous hops vector. For the base case $h=0$, observe that 
\[
\ph_v(v)=\{u\in N(v): w(\{u,v\})=\dist(u,v)\}.
\] 
It follows that, for every $u\in \ph_v(v)$, the value $d$ in the tuple $(t,d,s,b)$ sent by $u$ to $v$ at phase~$2$ satisfies $d=\dist(u,v)$. Therefore, the equality $d-w(\{u,v\}) = 0 = D[v]$ holds at~$v$, resulting in $u$ being added to $\varph[t]$, as desired. Now, let $h>0$, and let us assume that, for every $t\in V$ with $\maxhop(v,t)<h$, $\varph[t]=\ph_v(t)$ after $h+1$ phases. Let $t\in V$ with $\maxhop(v,t)=h$, and let $u\in \ph_v(t)$. Since $u\in \ph_v(t)$, $\maxhop(u,t)\leq h+1$, for which it follows that the distance to~$t$ is correctly set at~$u$ after $h+1$ phases. Therefore, at phase $h+2$, the value $d$ in the tuple $(t,d,s,b)$ sent by $u$ to $v$ satisfies $d=\dist(u,t)$, and thus the equality $d-w(\{u,v\}) = D[t]$ holds at~$v$, resulting in $u$ being added to $\varph[t]$. This completes the proof of the induction  for the previous hops vector.  

The proof completes by noticing that, as $D[t]=\dist(v,t)$ remains stable after $\minhop(v,t)$ phases, a node~$u$ added to $\varnh[t]$ or to $\varph[t]$ after the due number of phases stays in this set for the remaining phases of the algorithm.  
\end{proof}

\medskip
We now claim that, in Algorithm~\ref{algoBC}, every node $v\in V$ computes $\sigma_{v,t}$ and $\bc_v(t)$ for every node $t\in V$, at Lines~\ref{lineBCbegin}-\ref{lineBCend}. We treat $\sigma_{v,t}$ and $\bc_v(t)$ separately, as the former is somehow computed top-down, while the latter is computed bottom-up. 

\medskip

\begin{lemma}\label{lem:sigma-equal-S}
Algorithm~\ref{algoBC} enables every node~$v\in V$ to compute $\sigma_{v,t}$ for every node $t\in V$.  More specifically, at node~$v$, for every node~$t\in V$, we have $S[v,t]=\sigma_{v,t}$ after $\Diam(G)+1$ phases. 
\end{lemma}

\begin{proof} 
Let $v\in V$. The proof is by induction on $h=\maxhop(v,t)$. The statement holds for $h=0$, i.e., for $t=v$, as $S[v,v]=1$ in the function \textsc{init}, and $S[v,v]$ is not modified by the function \textsc{receive} (cf. the test performed at Line~\ref{line:computeS}). Now, let $h>0$, and assume that for every $t$ such that $\maxhop(v,t)<h$, $S[v,t]=\sigma_{v,t}$ after $h$ phases. Consider phase $h+1$. By Lemma~\ref{lem:NHetPH}, we have $\varnh[t]=\nh_v(t)$ at Line~\ref{line:computeS}. Moreover, by induction, since $\maxhop(x,t)<\maxhop(v,t)$ for every $x\in \nh_v(t)$, we have $S[x,t]=\sigma_{x,t}$ for every such node~$x$ after at most $h$ phases. During phase $h+1$, node~$v$ receives all the values $s=S[x,t]$ from these nodes, and thus, once it has received all of them, the update of Line~\ref{line:computeS} yields $S[v,t]=\sigma_{v,t}$, by Fact~\ref{fact:sigmaphs}.
\end{proof}

\begin{lemma}\label{lem:BCin2Dtime}
Algorithm~\ref{algoBC} enables every node~$v\in V$ to compute $\bc_v(s)$ for every node $s\in V\setminus\{v\}$.  More specifically, at node~$v$, for every node~$s\in V\setminus\{v\}$, $B[v,s]=\bc_v(s)$ after $2  \Diam(G)+1$ phases. 
\end{lemma}

\begin{proof} 
Let $v\in V$. The proof is by induction on $h=\Diam(G)-\maxhop(v,s)$, that is, we show that, for every~$s$ such that $\Diam(G)-\maxhop(v,s)\leq h$, $B[v,s]=\bc_v(s)$ after $\Diam(G)+h+2$ phases. 

For the base case $h=0$, let $s\in V$ such that $\maxhop(v,s)=\Diam(G)$. (If there are no such $s$, then the base case  holds trivially for~$v$). By Lemma~\ref{lem:NHetPH}, after $\Diam(G)+2$ phases, we have $\varph[s]=\ph_v(s)$ at Line~\ref{lineBCend}. Therefore, $\varph[s]=\emptyset$, because $\ph_v(s)=\emptyset$  as $v$ is a node at maximum hop-distance from~$s$. As a consequence, $B[v,s]$ is correctly set to $0=\bc_v(s)$ at Line~\ref{lineBCend}. 

Now, let $h\geq 0$, and let us assume that, for every $s$ such that $\Diam(G)-\maxhop(v,s)\leq h$, $B[v,s]=\bc_v(s)$ after $\Diam(G)+h+2$ phases. Let $s$ such that $\Diam(G)-\maxhop(v,s) = h+1$. (Again, if there are no such $s$, then the induction step $h\to h+1$ holds trivially for~$v$). By Lemma~\ref{lem:NHetPH}, after $\Diam(G)+2$ phases, we have $\varph[s]=\ph_v(s)$ at Line~\ref{lineBCend}. Moreover, by Lemma~\ref{lem:sigma-equal-S}, we have $S[v,s]=\sigma_{v,s}$ and, for every $x\in \ph_v(s)=\varph[s]$, the equality $S[x,s]=\sigma_{x,s}$ holds. Furthermore, for every $x\in\ph_v(s)$, we have $\maxhop(x,s)>\maxhop(v,s)$, from which it follows by induction that $B[x,s]=\bc_x(s)$ after $\Diam(G)+h+2$ phases. It then follows from Lemma~\ref{lem:recursive-formula-bc} that the computation performed at Line~\ref{lineBCend} guarantees that $B[v,s]=\bc_v(s)$, which completes the induction step.
As we only need to compute $\bc_v(s)$ for $s\neq v$, we need only to consider $h<\Diam(G)$, and the lemma follows.
\end{proof}

Since $\bc_v$ does not depend on $\bc_v(v)$ but only on $\bc_v(s)$ for $s\in V\setminus\{v\}$, Lemma~\ref{lem:BCin2Dtime} immediately implies that Algorithm~\ref{algoBC} enables every node~$v\in V$ to compute $\bc_v$ at Line~\ref{lineBC}, after $2\Diam(G)+1$ phases, up to renormalization by $1/((n-1)(n-2))$. The following theorem summarizes the results in this section. 
\begin{theorem}
Algorithm~\ref{algoBC} enables every node to compute its betweenness  centrality in any network $G$ after $2  \Diam(G)+1$ phases. 
\end{theorem}

\begin{figure}
 \centering
 \includegraphics[scale=0.2]{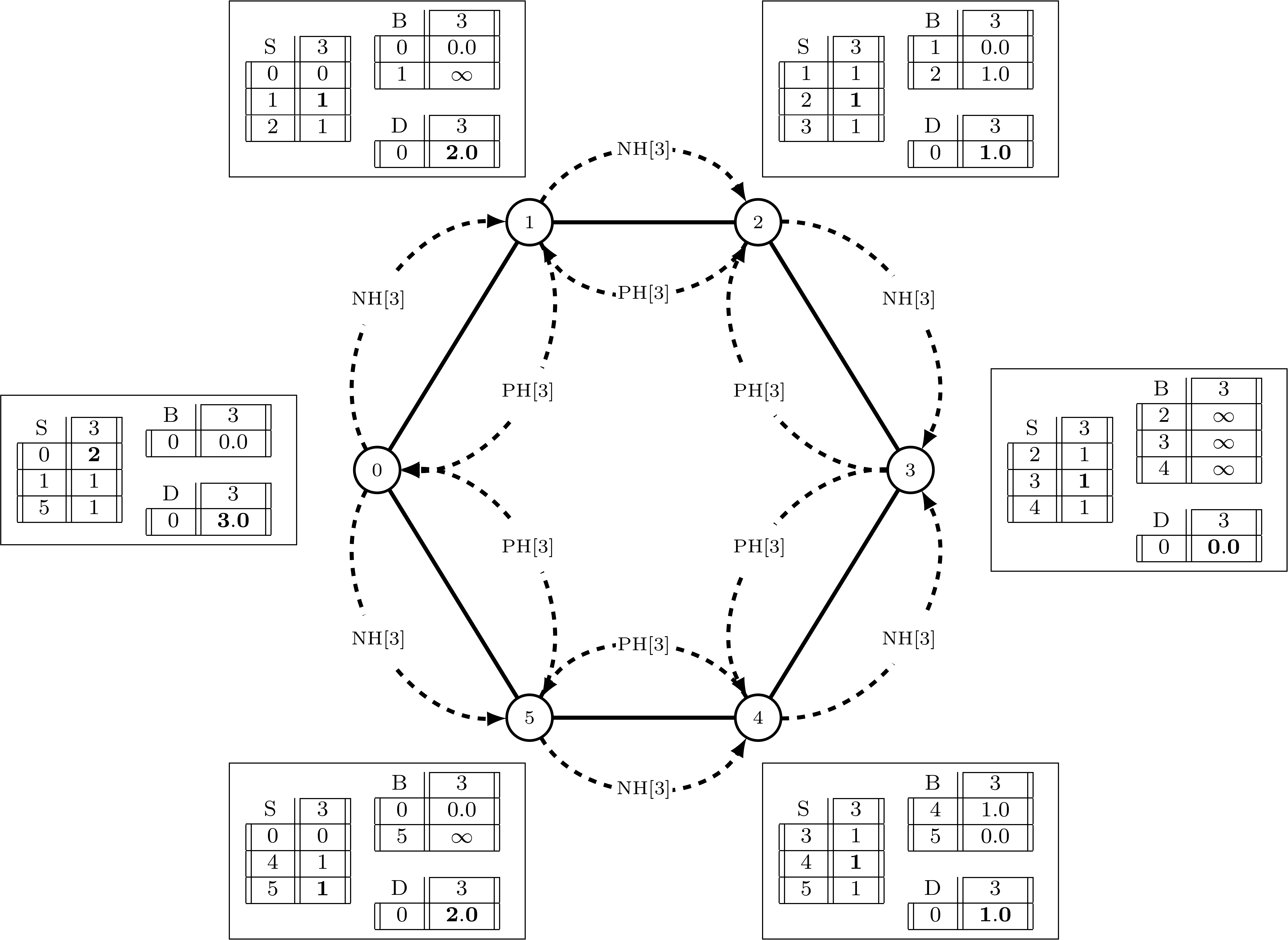}
 \caption{The state of Algorithm~\ref{algoBC} executed on a cycle of $6$ nodes, at the end of the fourth send/receive phase of the algorithm.}
 \label{fig:execution04}
\end{figure}

\subsection{An example of execution of Algorithm~\ref{algoBC}}

We now describe few steps of the execution of Algorithm~\ref{algoBC} on an unweighted cycle of six nodes, whose diameter is thus $\Diam=3$. In particular, we show how the algorithm computes the contribution of node $t=3$ to the betweenness centrality of all nodes, after the first $4$ phases have been completed (that is, after the correct value of the sets $\varnh$ and $\varph$ and of the matrix $\mathrm{S}$ has been computed). In Figure~\ref{fig:execution04} we show the state of all nodes at this moment of the algorithm execution. For example, node~$5$ now knows that its set of next hops towards node~$3$ contains only node~$4$, while its set of previous hops towards node~$3$ contains only node~$0$ (which is the only neighbor of node~$5$ which admits a shortest path towards node~$3$ passing through node~$5$). Note that all nodes have also correctly computed their distance from node~$3$ (in the case of node~$5$, for example, this distance is $2.0$), and the number of shortest paths connecting them to node~$3$ (in the case of node~$5$, for example, this number is $1$). In the figure, all these correct values are shown in boldface.

\begin{figure*}
 \centering
 \includegraphics[scale=0.2]{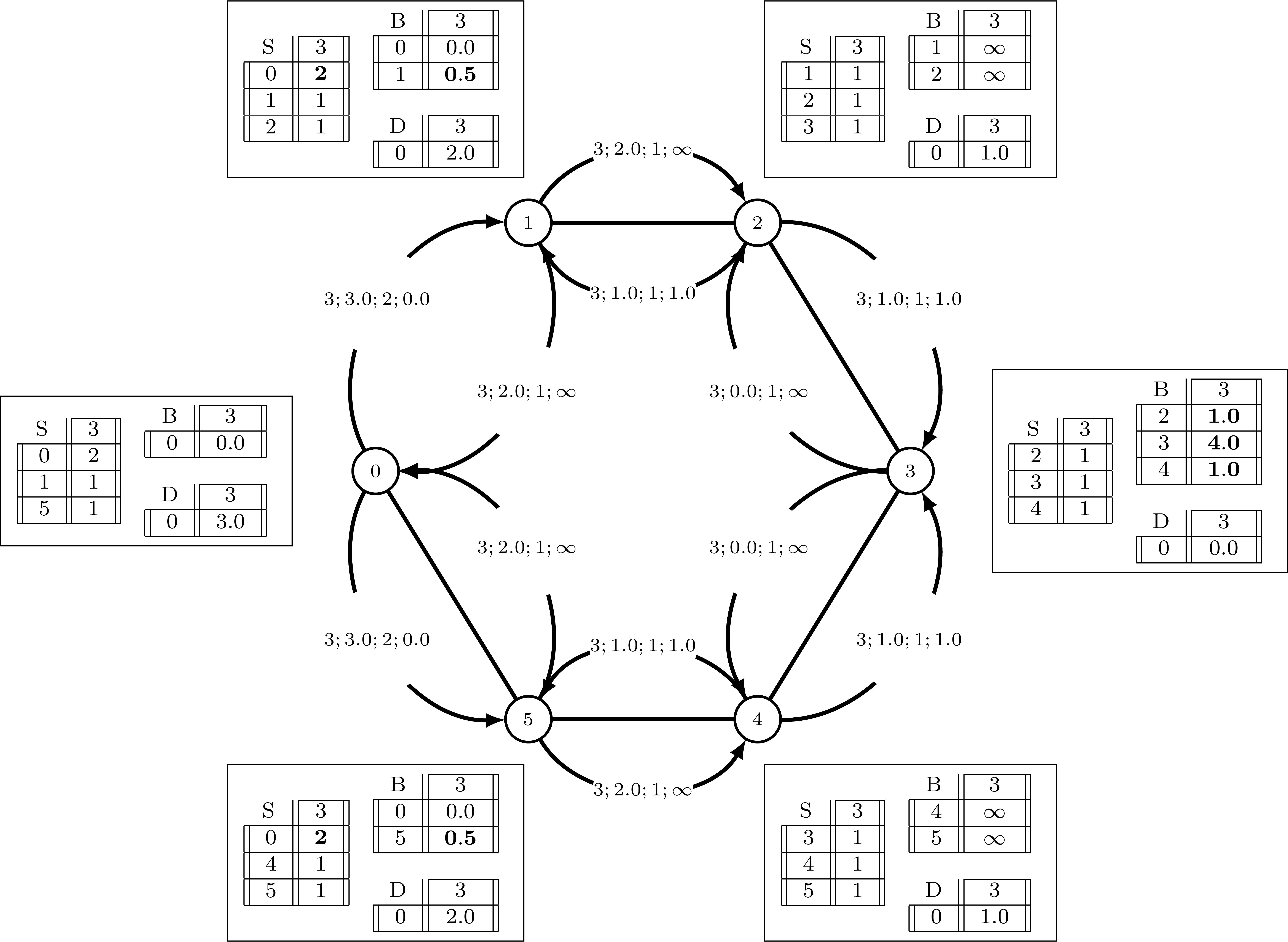}
 \caption{The state of Algorithm~\ref{algoBC} executed on a cycle of $6$ nodes, at the end of the fifth send/receive phase of the algorithm.}
 \label{fig:execution05}
\end{figure*}

In Figure~\ref{fig:execution05}, we show the messages sent by each node to its neighbors concerning node $t=3$ in the fifth send/receive phase (these are the messages sent by the function \textsc{send} defined in lines~11--14 of Algorithm~\ref{algoBC}). For instance, node~$5$ sends to both node~$0$ and node~$4$ the message $(3,2.0,1,\infty)$, telling them that its distance from node~$3$ is $2.0$, that there is only one shortest path from it to node~$3$, and that the current contribution of node~$3$ to its betweenness centrality is still unknown. On the other hand, suppose that node~$5$ receives first the message $(3,3.0,2,0.0)$ from node~$0$ and then the message $(3,1.0,1,1.0)$ from node~$4$. Even if node~$0$ is taken out of $\varph[3]$ by node~$5$ at line~16 of Algorithm~\ref{algoBC}, this node is reinserted in $\varph[3]$ by node~$5$ at line~23 of Algorithm~\ref{algoBC}, since the distance of node~$5$ from node~$3$ (that is, the value $\mathrm{D}[3]$) is equal to $d=3.0$ minus $1$ (remember that the graph is unweighted). Hence, the set $\varph[3]$ of node~$5$ does not change. After receiving the message $(3,3.0,2,0.0)$ from node~$0$, node~$5$ updates (at lines~24--25 of Algorithm~\ref{algoBC}) $\mathrm{S}[0,3]$ (which becomes $2$) and $\mathrm{B}[0,3]$ (which remains $0.0$). Since the set $\varnh[3]$ contains only node~$4$, the value $\mathrm{S}[5,3]$ is then set equal to $\mathrm{S}[4,3]=1$ (line~26 of Algorithm~\ref{algoBC}). At line~27 of Algorithm~\ref{algoBC}, the value $\mathrm{B}[5,3]$ is set equal to $1\cdot\frac{0+1}{2}=0.5$. Note that, at this moment, node~$5$ has correctly computed the number of shortest paths connecting node~$0$ to node~$3$, and the contribution of node~$3$ to its betweenness centrality (these values are shown in bold in the figure). It is easy to verify that the arrival of the message from node~$4$ does not change the state of node~$5$. Similarly, node~$1$ (which is symmetric to node~$5$) and node~$3$ update their state (once again, the updates are shown in bold in the figure).

\begin{figure*}
 \centering
 \includegraphics[scale=0.2]{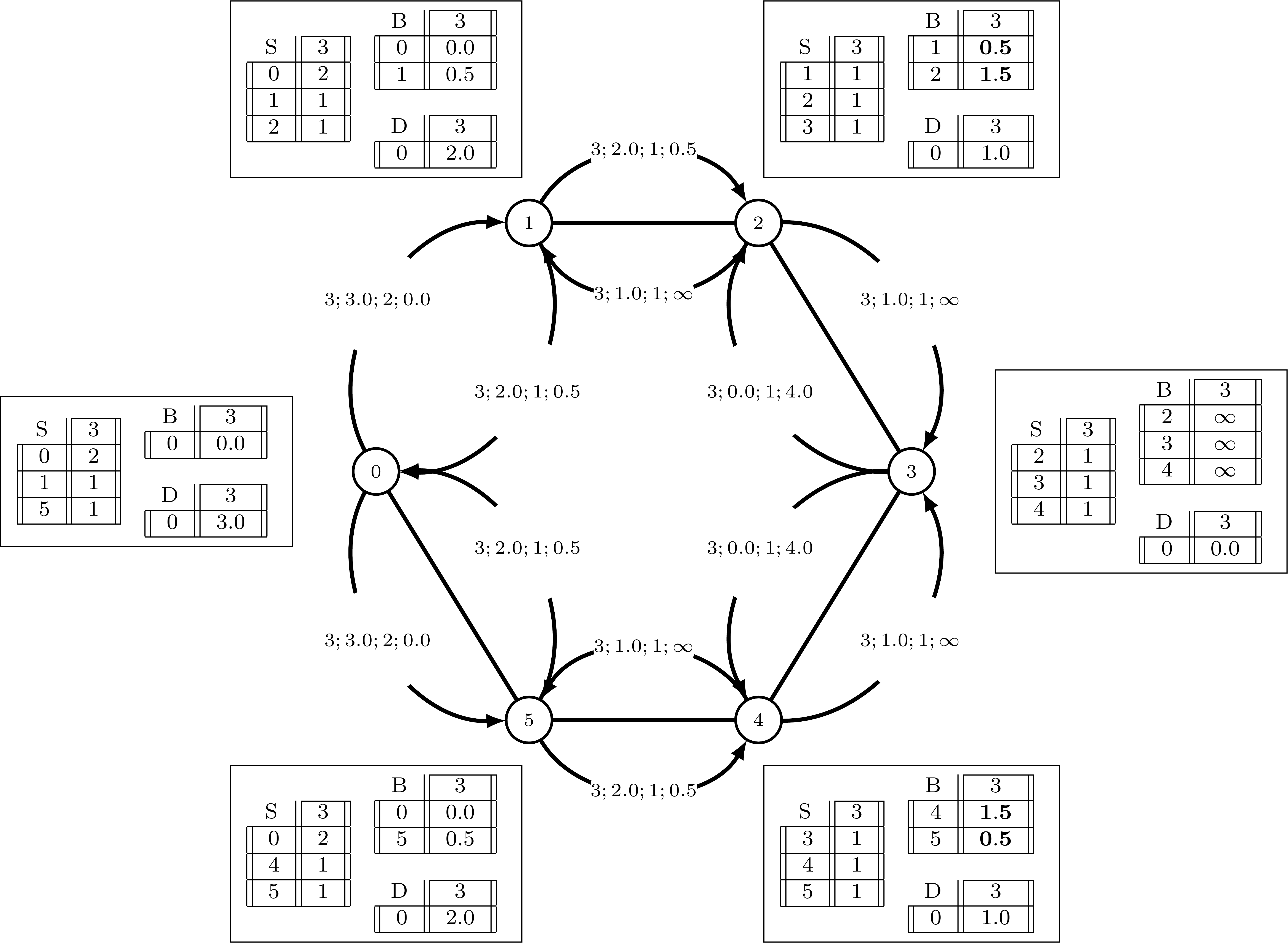}
 \caption{The state of Algorithm~\ref{algoBC} executed on a cycle of $6$ nodes, at the end of the sixth send/receive phase of the algorithm.} 
 \label{fig:execution06}
\end{figure*}

The execution of the next two phases of Algorithm~\ref{algoBC} are shown in Figure~\ref{fig:execution06} and~\ref{fig:execution07}, respectively. In particular, at the end of the seventh phase, each node $v$ has correctly computed, for each of its neighbors $u$, the number of shortest path connecting $u$ to node~$3$, and the contribution of node~$3$ to the betweenness centrality of $v$. 

\begin{figure*}
 \centering
 \includegraphics[scale=0.2]{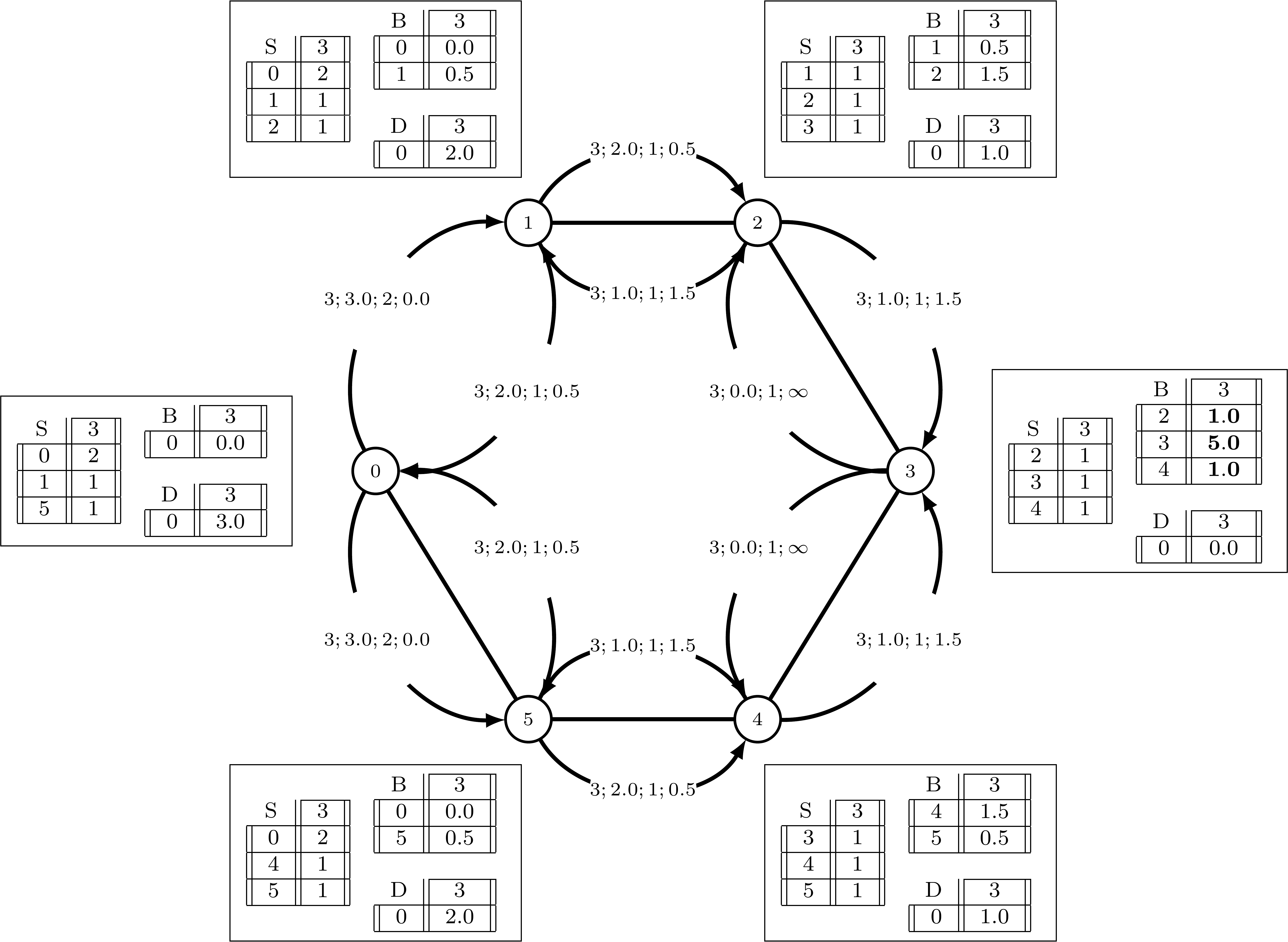}
 \caption{The state of Algorithm~\ref{algoBC} executed on a cycle of $6$ nodes, at the end of the seventh and last send/receive phase of the algorithm.}
 \label{fig:execution07}
\end{figure*}

\section{Implementation of Algorithm~\ref{algoBC} }

In this section, we present a more practical implementation of Algorithm~\ref{algoBC}, in which every node performs only a constant number of elementary operations upon reception of every message. The instructions performed at Lines~\ref{line:computeS} and~\ref{lineBCend} of Algorithm~\ref{algoBC} may, indeed, consume a large time, if $\nh_v(t)$ or $\ph_v(t)$ are large (note that the operations on the array $\varnh$ and $\varph$ can be implemented in (amortized) constant time using appropriate data structures). 

For this reason, we modify the function \textsc{receive} in order to accumulate values for the computation of $\sigma_{v,t}$ and $\bc_v(t)$ instead of summing up a potentially large number of values. This is done by the modified \textsc{receive} function described in Algorithm~\ref{algoBCfast}. In addition, the variable $C$ must be initialized to~$0$ in \textsc{init}, while the auxiliary array $A$, which stores the contribution of $u\in N(v)$ to $\bc_v(t)$, does not need to be initialized.

\begin{algorithm}[tb]
\caption{Implementation of Function \textsc{receive} of Algorithm~\ref{algoBC}, with constant number of elementary operations per message: instructions performed at node~$v$}
\label{algoBCfast}
\begin{algorithmic}[1]

\Function{receive}{message $(t,d,s,b)$ from $u\in N(v)$} 
\State \textbf{if} $t\neq v$ \textbf{then} $C \leftarrow C - B[v,t]$ \label{line:initC}
\If {$u\in \varnh[t]$}  \label{line:testNH2}
	\State $\varnh[t]\leftarrow \varnh[t] \smallsetminus \{u\}$ 
	\State \textbf{if} $t\neq v$ \textbf{then} $S[v,t]\leftarrow S[v,t]-S[u,t]$ \label{line:removeS}
\EndIf

\If {$u\in \varph[t]$}   \label{line:testPH2}
	\State $\varph[t]\leftarrow \varph[t] \smallsetminus \{u\}$ 
	\State $B[v,t]\leftarrow B[v,t]-A[u,t]$  \label{line:removeB}
\EndIf

\State $S[u,t] \leftarrow s$ 
\State $B[u,t] \leftarrow b$

\If{$d+w(\{u,v\})<D[t]$}
	\State $D[t]\leftarrow d+w(\{u,v\})$ 
\ElsIf{$d+w(\{u,v\}) = D[t]$} 
	\State $\varnh[t]\leftarrow \varnh[t] \cup \{u\}$ 
	\State \textbf{if} $t\neq v$ \textbf{then} $S[v,t]\leftarrow S[v,t]+S[u,t]$ \label{line:addS}
\ElsIf{$d - w(\{u,v\}) = D[t]$}
	\State $\varph[t]\leftarrow \varph[t] \cup \{u\}$ 
	\If{$S[u,t]\neq 0$} 
		\State $A[u,t] \leftarrow S[v,t] \cdot \frac{B[u,t]+1}{S[u,t]}$ 
	\Else 
		\State $A[u,t] \leftarrow 0$
	\EndIf
	\State $B[v,t]\leftarrow B[v,t]+A[u,t]$  \label{line:addB}
\EndIf 

\State \textbf{if} $t\neq v$ \textbf{then} $C \leftarrow C+B[v,t]$ \label{line:updateC}
\EndFunction 
\end{algorithmic}
\end{algorithm}

Algorithm~\ref{algoBCfast} aims at replacing the sums in the instructions
\begin{align*}
S[v,t] 	& \leftarrow 	 \sum_{x\in \varnh[t]}S[x,t] \\
B[v,t]		& \leftarrow 	 S[v,t]\cdot \sum_{x\in \varph[t]}\frac{B[x,t]+1}{S[x,t]}\\ 
C 		& \leftarrow 	 \sum_{x\neq v} B[v,x]
\end{align*}
in Algorithm~\ref{algoBC} by a bounded number of operations. The instruction $S[v,t] \leftarrow \sum_{x\in \varnh[t]}S[x,t]$ is replaced by removing $S[u,t]$ from $S[v,t]$ whenever $u$ must be removed from $\varnh[t]$ (cf. Line~\ref{line:removeS}), and adding $S[u,t]$ to $S[v,t]$ whenever $u$ must be added to $\varnh[t]$ (cf. Line~\ref{line:addS}). Similarly, the instruction $B[v,t] \leftarrow S[v,t]\cdot \sum_{x\in \varph[t]}\frac{B[x,t]+1}{S[x,t]}$ is replaced by removing from $B[v,t]$ the previously added contribution of $u$ to $B[v,t]$ stored in $A[u,t]$ (cf. Line~\ref{line:removeB}), whenever $u$ must be removed from $\varph[t]$, and adding the contribution $S[v,t]\cdot \frac{B[u,t]+1}{S[u,t]}$ to $B[v,t]$ whenever $u$ must be added to $\varph[t]$ (cf. Line~\ref{line:addB}), after having saved it into $A[u,t]$. Finally, the instruction $C \leftarrow \sum_{x \neq v} B[v,x]$ is replaced by exchanging the old value of $B[v,t]$ with the (potentially) new value computed at Line~\ref{line:addB} --- cf. Lines~\ref{line:initC} and~\ref{line:updateC}. 

\begin{figure*}
    \centering
    \includegraphics[scale=0.17]{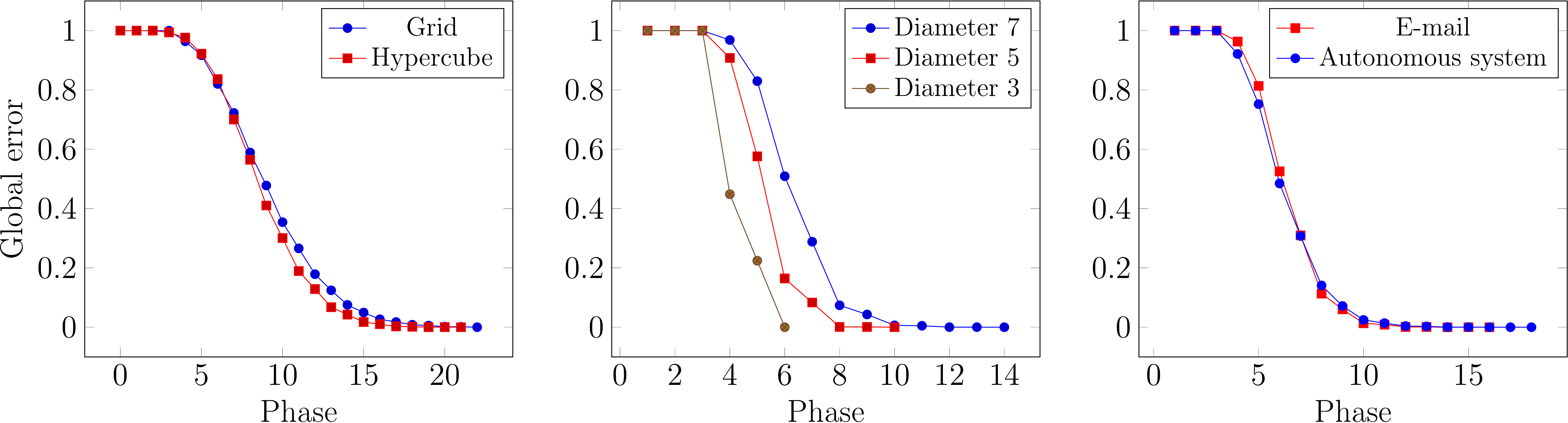}
    \caption{The (mean) global error as a function of time in (left) a $7\times 6$ grid and a hypercube of dimension 11 (where, hence, diameter is $11$), (center) Erd\"{o}s-Renyi graphs with 500 nodes and different diameters (20 samples for each diameter), and (right) an e-mail network with 1133 nodes and in an autonomous system network with 3011 nodes.}
    \label{fig:convergence}
\end{figure*}

\section{Experimental Results}

We have implemented a Java simulator in order to analyze the global and the local convergence times, when applying  Algorithm~\ref{algoBCfast} to lattice networks, to networks generated by different random graph models, and to real-world networks. The simulator uses a virtual clock and,  for each node, triggers a send-event once per virtual time unit (also called \textit{phase}). Unlike~\cite{MaccariGGMC18}, no random jitter has been added to the events scheduled by nodes, so the simulator analyzes the algorithm in the case of perfect synchronization. Each simulation ends when all nodes converge to steady state, that is, when the $C$ values do not change anymore. In order to perform the global convergence study, during each simulation we also record the phases at which the $D$ and the $S$ values do no change anymore. Finally, for the local convergence study we also record, for each node $v$, the last phase in which the $C$ value of $v$ changed. The output of our algorithms is compared with the output of the Python NetworkX library~\cite{networkx}.

\subsection{The datasets}

We have performed our convergence study on the following networks (in the rest of this section, we will present the results relative only to some of these networks, since the results on the others are very similar).

\begin{itemize}

\item Hypercubes of several different dimensions, grids with several different widths and heights, and complete binary trees of different heights.

\item Random Erd\"{o}s-R\'enyi graphs with different diameters~\cite{Erdos1959}, random Barab\'asi-Albert graphs with different numbers of links for each new node added to the graph~\cite{Albert2002}, and random geometric graphs with different communication ranges~\cite{Gilbert1961}.

\item Several real-world networks, such as an e-mail graph, whose edges indicate e-mail interchanges between members of the Univeristy Rovira i Virgili (Tarragona, Spain)~\cite{Guimera2003}; several autonomous system networks, which are communication networks of ``who-talks-to-whom'' obtained from the BGP logs~\cite{snapnets}; a road graph, which contains a large portion of the road network of the city of Rome, Italy, from 1999 (vertices correspond to intersections between roads, edges correspond to roads or road segments, and weights correspond to distances)~\cite{DIMACS2006}; and a co-authorship (both unweighted and weighted) graph~\cite{Newman2001}.
\end{itemize}

\subsection{Analysis of global convergence time}
We start by studying the global convergence over time. To this end, at each phase $P$, we study the $\ell_2$-norm of the distance between the current betweenness centrality computed by the nodes, $C=C(P)$, and the actual betweeenness centrality, $\bc$. This is normalized by the actual betweeenness centrality, which gives the following global-error formula:
\[\frac{\|\bc-C\|_2}{\|\bc\|_2}=\frac{\sqrt{\sum_{v\in V}(\bc_v-C[v])^2}}{\sqrt{\sum_{v\in V}(\bc_v)^2}}\]

\subsubsection{Unweighted networks: lattices, Erd\"{o}s-Renyi, e-mail}
Our first experiments are with unweighted, synthetic networks.
We start with a grid and a hypercube (Figure~\ref{fig:convergence} (left)), both with diameter~11.  Grids and hypercubes with different dimensions present very similar behaviours, and so do binary trees. In the first $3$ phases, the betweenness centrality of all nodes is $0$, since the list $\varph$ is empty. This phenomena is normal and predictable, and reoccurs in all our experiments, both in unweighted and weighted graphs.

As predicted by our analysis, the algorithm converges on both networks in  roughly $2\Diam(G)$ time. 
In an unweighted network, the estimate of $\bc$ made by each node can only increase: as time passes, each node learns about more shortest paths it belong to, increases its estimate of the betweenness centrality, and thus makes it more accurate, until converging to the right values. 
In these examples, the values computed by the nodes after $\Diam+2$ phases already give a relatively low error (less than~$10\%$).

Next, we study the convergence on Erd\"{o}s-Renyi graphs.
We average over 20 randomly-generated graphs with a given diameter, for diameters $3,5$ and $7$ (Figure~\ref{fig:convergence} (center)).
While slightly more cluttered, the results are similar to the ones observed in the lattices above. The betweenness centrality is not updated in the first 3 phases, and then rapidly and monotonically decreases, until full convergence in $2\Diam$ phases. Other random graphs present similar behavior.

The last two convergence analyses of unweighted graphs are the ones of the e-mail network and of the autonomous system network (Figure~\ref{fig:convergence} (right)). The convergence patterns are of similar nature, and are actually slightly more rapid then in the previous experiments --- the convergence time is still $2\Diam+1$, but the error after $\Diam$ phases is smaller.

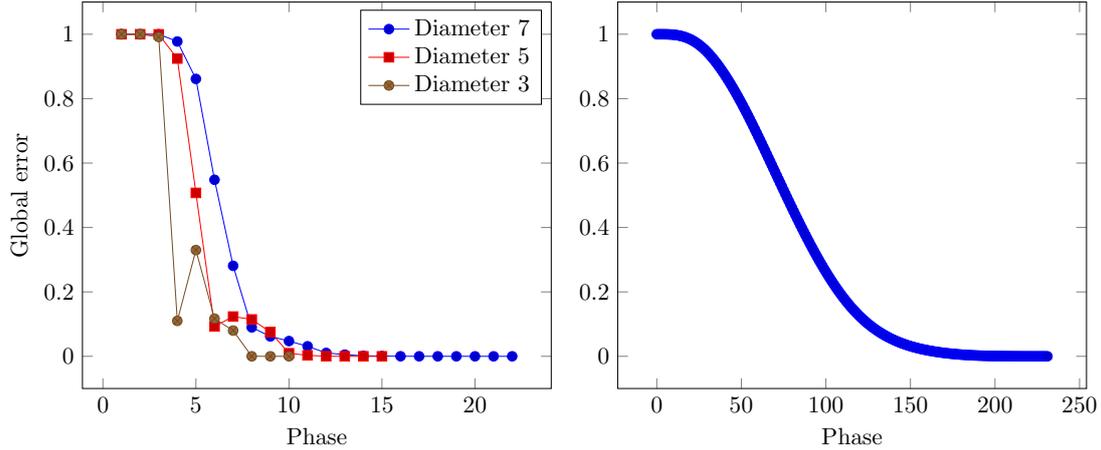
\begin{figure}[t]
    \renewcommand{\tikzscale}{0.9}
		\centering
		\input{img_conv_wer_20_500.tex} \input{img_conv_rome.tex}
		\caption{The (mean) global error as a function of time in randomly weighted Erd\"{o}s-Renyi with 500 nodes (20 samples per diameter, where the noted diameter is the one of the underlying, unweighted graphs), and in a road network with 3353 vertices.}
	\label{fig:convergence_weighted}
\end{figure}

\subsubsection{Weighted networks: Erd\"{o}s-Renyi, road network}

We now turn to describe our experiments with weighted Erd\"{o}s-Renyi graphs (Figure~\ref{fig:convergence_weighted} (left)). We build the graphs as before, and then assign random weights form the set $\{1,2,5\}$ to each edge, with expected weight $2$. The shapes of the convergence curves are generally similar to the unweighted case, with one interesting difference: the global error is no longer monotonically decreasing, and, instead, it sometimes increases, but always before $\Diam$ phases. This phenomena is explained by the fact that the algorithm first finds paths of the least number of hops, but then updates them to paths with more hops but less weight. When such an update occurs, $S[v,t]$ may decrease, as many few-hops paths are replaced by a few (or one) lighter paths. This causes $B[v,t]$ to temporarily decrease, and the error to increase. This phenomena is moderated when the diameter grows larger.

A real world, large scale weighted network we study is composed of a large portion of the road network of Rome, represented by a weighted graph (Figure~\ref{fig:convergence_weighted} (right)). In this network, whose diameter $\Diam(G)$ is very high (roughly 120 hops), we observe a very smooth convergence pattern, reminding the clean convergence patterns of the unweighted lattices (Figure~\ref{fig:convergence} (left)).
This could be an artifact of the large diameter or of the larger number of different weights.

\subsection{Analysis of local convergence time}
A different perspective on the convergence time is given by considering the number of nodes which converged completely over time. 
We consider two values: $T_D$, the time it takes for the Bellman-Ford algorithm to converge locally, i.e., until the distances are correctly computed; and, $T_C$, the time it takes for the betweenness centrality value to converge.
A third value of interest is the convergence time of $S[v,t]$, the number of shortest path. However, in unweighted graphs our algorithm always computes this value exactly one round after the convergence of the length of the shortest paths ($T_D$). In weighted graphs, experiments show  that this is almost always true as well, so we do not include this value in our plots.

\subsubsection{Unweighted networks: Erd\"{o}s-Renyi and e-mail}
We study $T_D$ and $T_C$ for all the nodes of three unweighted Erd\"{o}s-Renyi graphs (Figure~\ref{fig:er-scatterplot}), and for the e-mail and the autonomous system networks (Figure~\ref{fig:email-as-scatterplot}).
All these networks present a somewhat similar behaviour. 
As predicted, all distances converge after $\Diam$ phases, and the betweenness centrality values converge after $2\Diam+1$ phases. 
The distribution of the convergence times of $T_D$ reassembles a Poisson distribution: most of the nodes, and especially nodes with a large betweenness centrality, correctly compute all their distances relatively fast --- these are central nodes.
Few, peripheral nodes, take more time to compute their distances correctly, and these nodes have lower betweenness centrality.
The \emph{eccentricity} of a node is the maximal distance from it to all other nodes, and it is always  between $\diam/2$ and $\diam$. The time $T_D$ it takes a node to find all its distances to other nodes is thus exactly its eccentricity (up to a small additive constant), which better explains why $T_D$ is always between $\diam/2$ and $\diam$. Note that in unweighted graphs, $\Diam=\diam$.

\begin{figure}[t]
    \centering
    \includegraphics[scale=.95]{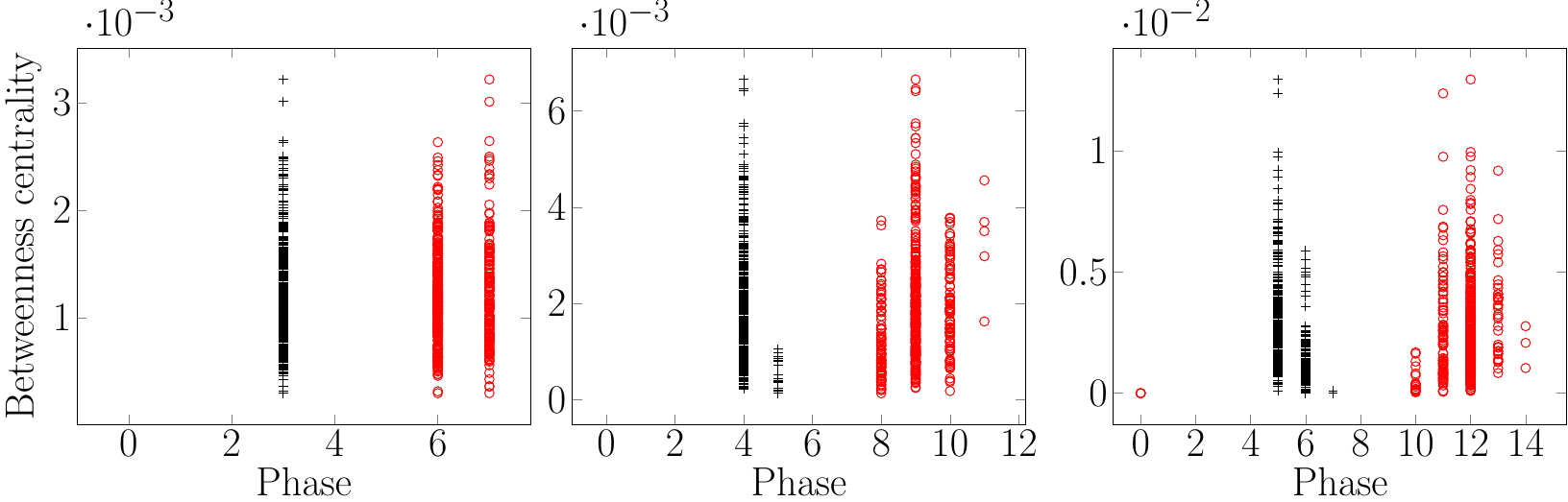}
    \caption{$T_D$ and $T_C$ for all the nodes of three Erd\"{o}s-Renyi graphs with diameter 3, 5, and 7, respectively.}
    \label{fig:er-scatterplot}
\end{figure}

In the case of the Erd\"{o}s-Renyi graphs, when the random network has small diameter (3 and 5 in our examples) and is very dense, each node lies at least on one shortest path, and they all have non-zero betweenness centrality. In networks of larger diameter (Erd\"{o}s-Renyi with diameter 7, e-mail network and autonomous system network), we notice that some nodes do not lie on any shortest path --- these nodes' initial estimate of their betweenness centrality, $0$, is correct, and thus they also have $T_C=0$, appearing in the lower-left corner of the figures.

The convergence time of $T_C$ reassembles a normal distribution: nodes with low betweenness centrality may take longer or shorter times than nodes with very high betweenness centrality to compute their betweenness centrality values. 
This phenomena can also be explained by studying the eccentricity of the nodes: we have found that in general, nodes with large betweenness centrality have low eccentricities, of roughly $\Diam/2$ which causes faster convergence.
For example, in the case of the autonomous system network, the average eccentricity of all nodes is approximately~$6.6$, while the average eccentricity of the ten nodes with highest betweenness centrality is roughly~$5.5$ (note that this network has diameter $9$, so no node can have eccentricity smaller than $5$).

\begin{figure}
    \centering
    \includegraphics[scale=0.25]{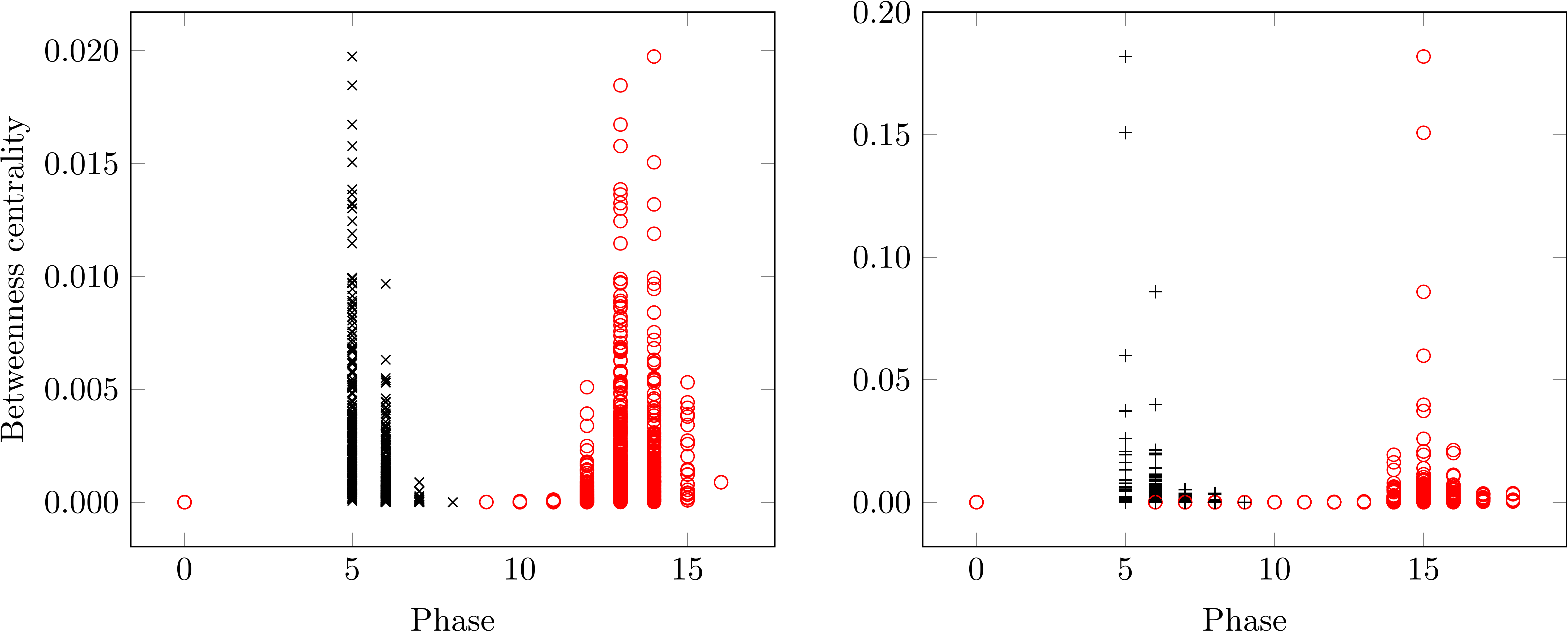}
    \caption{$T_D$ and $T_C$ for all the nodes of the e-mail network (left) and of the autonomous system network (right).}
    \label{fig:email-as-scatterplot}
\end{figure}

Thus, after the computation of all distances and numbers of shortest paths, which takes roughly $\Diam$ phases, nodes of high betweenness centrality take only $\Diam/2$ more phases to compute their betweenness centrality. Nodes with low betweenness centrality may have different eccentricities, which explains the wider scatter of their convergence times.
To better understand the betweenness centrality convergence times $T_C$, we also check the \emph{number of nodes} that converge in each phase, as shown, for the case of the autonomous system network, in Figure~\ref{fig:as-bar}. In this figure, we omit roughly $65\%$ of the nodes, that have $\bc=0$ and converge in $0$ phases. The rest of the nodes present a normal distribution, where most nodes converge in roughly $\frac{3}{2}\Diam$ phases.

\begin{figure}[h]
    \centering
    \includegraphics[scale=1]{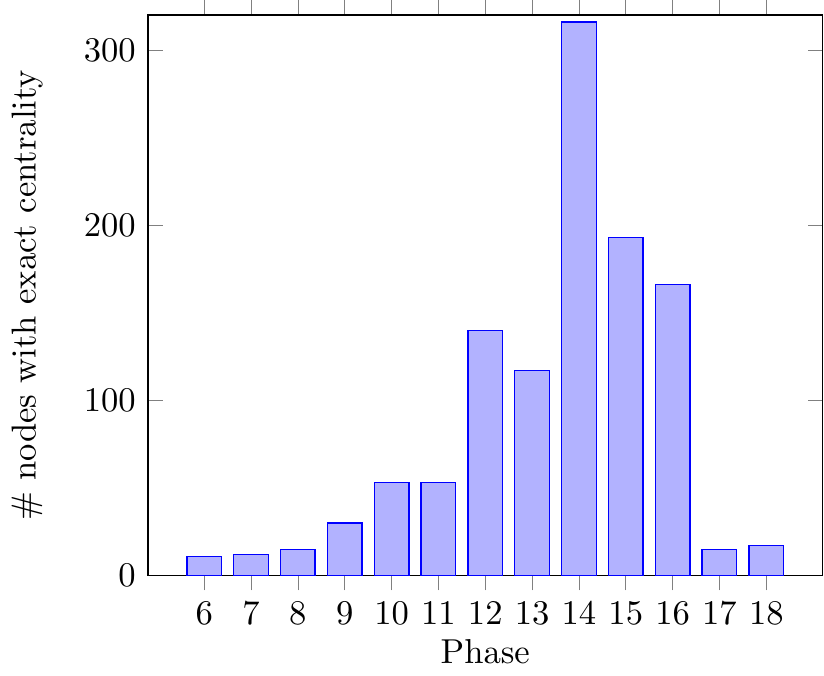}
    \caption{The number of nodes whose $C$ value converges to the exact centrality value as a function of the phase, in the case of the autonomous system network.}
    \label{fig:as-bar}
\end{figure}

\subsubsection{Weighted networks: road network}
Finally, we revisit the weighted graph representing the road network of Rome. The scatter of betweenness centrality values and convergence times of $T_D$ and $T_C$ give a distribution very similar to the ones of the unweighted networks (see Figure~\ref{fig:rome-scatt-bar} (left)). It can be observed that the convergence times of the distances, $T_D$, of nodes with high betweenness centrality is low --- roughly $\Diam/2$ --- due to their low eccentricity. Other node take also time proportional to their  eccentricity, which is between $\diam/2$ and $\diam$ (this is a graph theoretic fact).
The convergence time of the betweenness centrality, $T_C$, is again scattered around $\frac{3}{2}\Diam$, with nodes of high betweenness centrality indeed converging in roughly this time, while other nodes may also converge faster or slower.

\begin{figure}
    \centering
    \includegraphics[scale=0.82]{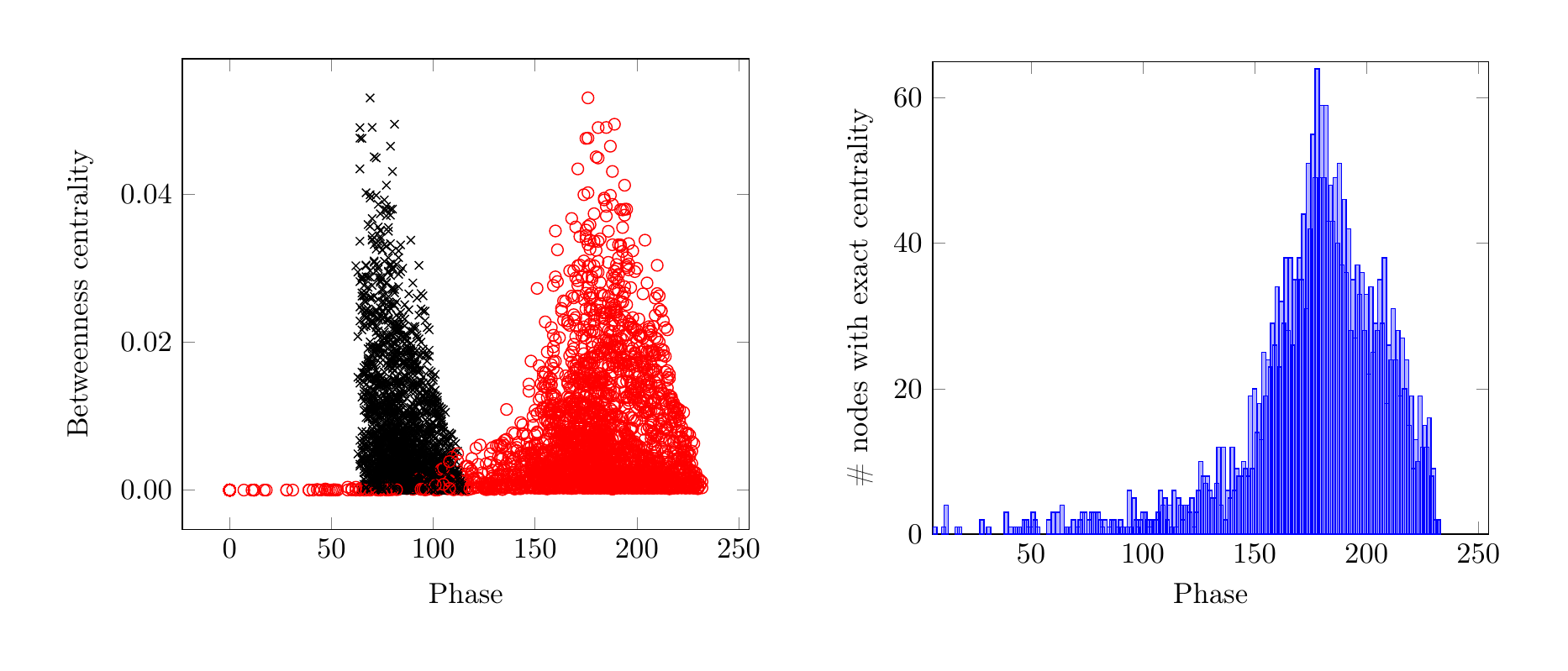}
    \caption{$T_D$ and $T_C$ for all the nodes (left) and number of nodes whose $C$ value converges to the exact centrality value as a function of time (right), in the case of the road network.}
    \label{fig:rome-scatt-bar}
\end{figure}

A unique phenomena observed here is that of nodes with betweenness centrality of~$0$, who's $T_C$ is \emph{not~$0$}. This is a result of paths going through these nodes, which are not shortest paths (in terms of weight), but seem to be so before enough hops are considered. All these nodes have their $T_C$ less than $\diam$, since in this time all shortest path distances are correct.

Once again, we also check how many nodes converge in each phase (see Figure~\ref{fig:rome-scatt-bar} (right)). Here, roughly half the nodes converge in $0$ rounds, with $\bc=0$. The rest of the nodes 
present a behaviour similar to before --- most of them take roughly $\frac{3}{2}\Diam$ phases to converge, while a few take longer or shorter, but never more than $2\Diam+1$.

\section{Conclusion}

In this paper, we have shown that betweenness centrality can be easily computed with only minimal modifications to the classical Bellman-Ford algorithm, and, hence, be integrated into distance-vector routing protocols. We have also presented an experimental analysis of the global and local convergence time of the proposed algorithm.

An interesting question left open by our work is designing bandwidth-adaptive distributed algorithms for betweenness centrality, assuming that the graph has a polynomial number of shortest paths (if this not the case, which would cause the $S$-values to have a linear number of bits, we could use a floating point representation with $O(\log n)$ bits to approximate the $S$-values, as described in~\cite{Hua2016}). Specifically, let $\congest(B)$ be the variant of the \congest\/ model described in Section~\ref{subsec:RW}, in which at most $B$ words of $O(\log n)$ bits each can be sent through each link at each round. The distributed algorithm for weighted graphs described in this paper applies to the $\congest(n)$ model, and converges in $O(\Diam)$ rounds. The distributed algorithms for unweighted graphs described in~\cite{HoangPDGYPR19,PR2018,HuaFAQLSJ2016} apply to the $\congest(1)$ model, and converge in $O(n)$ rounds. We thus leave open the question of computing betweenness centrality of weighted graphs in $O(n/B+\Diam)$ rounds in the $\congest(B)$ model, for every $1\leq B \leq n$.

\let\OLDthebibliography\thebibliography
\renewcommand\thebibliography[1]{
	\OLDthebibliography{#1}
	\setlength{\parskip}{0pt}
	\setlength{\itemsep}{0pt plus 0.3ex}
}
\bibliographystyle{plain}
\bibliography{bc}

\end{document}

%% file: img_conv_wer_20_500.tex
\begin{tikzpicture}[scale=\tikzscale, every node/.style={scale=\tikzscale}]
\begin{axis}[xlabel={Phase},ylabel={Global error}]
\addplot table{
1 1.0
2 1.0
3 0.9999295594709684
4 0.9775535310610618
5 0.8613227271752752
6 0.5480915517502833
7 0.28096191013559246
8 0.08978304694913726
9 0.06176359705060843
10 0.04735341151397893
11 0.031131868851746465
12 0.010584926223918881
13 0.005266229925090928
14 0.001566744792513643
15 7.308817278122513E-4
16 1.7961299587645617E-4
17 6.79682684135895E-5
18 1.6049376372173E-5
19 2.5714272298823894E-6
20 5.346995685704427E-8
21 1.4594047624050336E-8
22 0
};
\addplot table{
1 1.0
2 1.0
3 0.9996450470042676
4 0.9245610404109138
5 0.5073999596445494
6 0.0927653411362482
7 0.12347381473436277
8 0.11446097902966479
9 0.07547159946480825
10 0.009467307468918638
11 0.003088692920045613
12 1.9074554941137656E-4
13 3.603283345825361E-5
14 2.844548029480521E-15
15 2.8445194966306773E-15
};
\addplot table {
1 1.0
2 1.0
3 0.9914528840052613
4 0.11004444822604335
5 0.32986283597399585
6 0.11688105841785579
7 0.07988787343705513
8 7.458610708428173E-5
9 5.237240263010045E-6
10 3.152143236041534E-15
};
\legend{Diameter 7,Diameter 5,Diameter 3}
\end{axis}
\end{tikzpicture}

%% file: img_conv_rome.tex
\begin{tikzpicture}[scale=\tikzscale, every node/.style={scale=\tikzscale}]
\begin{axis}[xlabel={Phase}]
\addplot table{
0 1.0
1 1.0
2 0.999999989988012
3 0.9999805583840187
4 0.9999473110647462
5 0.9998660387374743
6 0.999749896008881
7 0.9995499273509799
8 0.9992883273584392
9 0.998903620726795
10 0.9984303556919242
11 0.9977954040339241
12 0.9970456487903111
13 0.9960967892870526
14 0.9950077240632043
15 0.9936839183796023
16 0.9921969859740033
17 0.9904393936606134
18 0.9884962456721951
19 0.9862511854072695
20 0.983805076098303
21 0.9810315871333538
22 0.9780480037487947
23 0.9747190534527325
24 0.971169250552361
25 0.9672616316222354
26 0.9631314856989206
27 0.9586342800021613
28 0.9539199668549353
29 0.9488331939197965
30 0.9435396010871231
31 0.9378777360845073
32 0.9320191303116642
33 0.925802119917978
34 0.9194111870233472
35 0.9126742442209222
36 0.9057795673950475
37 0.8985501929594267
38 0.8911908685789929
39 0.8835081469160643
40 0.8757155750908263
41 0.8676196228734973
42 0.8594271726989425
43 0.8509391877815552
44 0.8423667000420708
45 0.8335032328737546
46 0.8245604511825194
47 0.8153305746469383
48 0.8060343763810192
49 0.7964698911586755
50 0.7868568914824702
51 0.7769826245416445
52 0.7670662063250857
53 0.7569078764545922
54 0.7467187219337381
55 0.7362976359875005
56 0.7258646461071441
57 0.7152217435593304
58 0.7045864596937069
59 0.6937582771134538
60 0.6829571367772665
61 0.6719781163731982
62 0.6610426436733664
63 0.6499460818908637
64 0.6389021081888964
65 0.6277135081696078
66 0.6165907786204484
67 0.6053450633059294
68 0.5941839778562191
69 0.5829130482313021
70 0.5717403397969253
71 0.5604728574094029
72 0.549305111789621
73 0.5380503349755364
74 0.5269091977494869
75 0.5157035410797011
76 0.5046225784612457
77 0.4934901232113296
78 0.48249214703221366
79 0.47145996801045376
80 0.46057407188859295
81 0.44967117444466476
82 0.43892055431014615
83 0.42816328485623256
84 0.4175680801308486
85 0.4069766051831905
86 0.3965536059754112
87 0.38614939118695707
88 0.37592302867794847
89 0.3657282139753115
90 0.3557176782836445
91 0.3457514247202445
92 0.33597872259923045
93 0.32626045135925774
94 0.3167449069866389
95 0.307295465289693
96 0.29805100071859747
97 0.2888812216326318
98 0.27992243530194466
99 0.271049219515392
100 0.2623904176951968
101 0.2538250191147787
102 0.24547891004076314
103 0.2372359570875335
104 0.22921745070824753
105 0.22131259797694042
106 0.2136352379912872
107 0.20607643250175386
108 0.19874561966099552
109 0.19153799243870526
110 0.18455905609054674
111 0.1777092674810387
112 0.17108772334036998
113 0.16459475978777793
114 0.1583279681423258
115 0.15219378684355392
116 0.14628305962846433
117 0.1405055605590885
118 0.13494825727718182
119 0.12952370329088234
120 0.12431631850878651
121 0.11924028490683468
122 0.11437543122738979
123 0.10963656996025911
124 0.10509974676669383
125 0.1006825193469412
126 0.09645719728445198
127 0.09234586544009973
128 0.08841749242622833
129 0.08459707400361759
130 0.0809503651583182
131 0.07740544230704238
132 0.07402683470864492
133 0.07074431644540236
134 0.06762094070399946
135 0.06458782819594112
136 0.061703893332813405
137 0.05890253748102363
138 0.056241881595699615
139 0.053657469556696466
140 0.051204998658203466
141 0.04882264938481883
142 0.04656485012044648
143 0.04437169291027464
144 0.04229534312894827
145 0.04027865998571056
146 0.03837187969352738
147 0.036520404570245954
148 0.034771286641789335
149 0.03307261505072147
150 0.03147026575442329
151 0.02991398270380418
152 0.028447527061214978
153 0.027022536922336998
154 0.025680560628423537
155 0.024375975701516266
156 0.02314808507884207
157 0.021953469331146273
158 0.0208294928526983
159 0.01973475588685046
160 0.01870441531397645
161 0.017700333365403557
162 0.0167556039026418
163 0.015834261509621247
164 0.014967811762006985
165 0.014122712657995296
166 0.013328529603526276
167 0.012554241145432283
168 0.011827763174364833
169 0.01112005795052891
170 0.010457332423100775
171 0.009812554957169518
172 0.009210234070473349
173 0.008625101975418005
174 0.008080026926395879
175 0.007550896111704974
176 0.0070587459318671554
177 0.006580882551280887
178 0.006137546577973176
179 0.005707649788013504
180 0.005309633948786088
181 0.004923876982172576
182 0.004567264573349757
183 0.004221476794855157
184 0.0039024224818197274
185 0.003593472082509292
186 0.0033098046204630945
187 0.0030358310855969485
188 0.0027851760389057323
189 0.0025432454225250626
190 0.0023228136787870352
191 0.002110810934890829
192 0.0019196000708233169
193 0.0017368045907297936
194 0.0015738583334605379
195 0.0014187021378985216
196 0.0012812071204415351
197 0.0011506210496566513
198 0.001035690038287699
199 9.264328926357553E-4
200 8.3107568209387E-4
201 7.404608048002012E-4
202 6.616658406964696E-4
203 5.864960761125406E-4
204 5.210631726372802E-4
205 4.5830934739882694E-4
206 4.0402542581090565E-4
207 3.520040646211091E-4
208 3.0746813014591556E-4
209 2.6492310957959656E-4
210 2.287638358004362E-4
211 1.939833638577307E-4
212 1.6479706265441944E-4
213 1.367855980424849E-4
214 1.1382658800371637E-4
215 9.136780693042806E-5
216 7.312380515493414E-5
217 5.536733971381465E-5
218 4.216011825202302E-5
219 2.9956261788215304E-5
220 2.1972753640839512E-5
221 1.4829461418509392E-5
222 1.0773178165258488E-5
223 6.9350914661533766E-6
224 4.93911203597601E-6
225 2.834156535614241E-6
226 1.925182586034274E-6
227 8.479111129722255E-7
228 5.785197802857492E-7
229 1.5745314086745786E-7
230 1.1133618362650016E-7
231 9.665703362874488E-16
};
\end{axis}
\end{tikzpicture}